\documentclass[a4paper,12pt]{article}%
\usepackage{amsmath}
\usepackage{amssymb}
\usepackage{amsfonts}
\usepackage[onehalfspacing,nodisplayskipstretch]{setspace}
\usepackage{graphicx}%
\providecommand{\U}[1]{\protect\rule{.1in}{.1in}}

\newtheorem{theorem}{Theorem}

\newtheorem{proposition}[theorem]{Proposition}

\newenvironment{proof}[1][Proof]{\noindent\textbf{#1.} }{\ \rule{0.5em}{0.5em}}
\def\Rev{\normalfont\textsc{Rev}}

\newcommand\ignore[1]{}

\usepackage[dvipsnames]{xcolor}

\usepackage{tabularx} 
\usepackage{array}    

\begin{document}

\title{Existence of Optimal Mechanisms for Selling Multiple Goods: An Elementary
Proof\thanks{First version: June 2024.}}
\author{Sergiu Hart\thanks{%
The Hebrew University of Jerusalem (Federmann Center for the Study of
Rationality, Department of Economics, and Department of Mathematics).\quad 
\emph{E-mail}: \texttt{hart@huji.ac.il} \quad \emph{Web site}: \texttt{%
http://www.ma.huji.ac.il/hart}} \and Noam Nisan\thanks{%
The Hebrew University of Jerusalem (Federmann Center for the Study of
Rationality, and School of Computer Science and Engineering). Supported by a grant from the Israeli Science Foundation (ISF number 505/23). \emph{E-mail}: 
\texttt{noam@cs.huji.ac.il} \quad \emph{Web site}: \texttt{%
http://www.cs.huji.ac.il/\symbol{126}noam}}}

\maketitle

\begin{abstract}
        We provide an elementary proof that revenue-maximizing  mechanisms exist in multi-parameter settings whenever the distribution of valuations has finite expectation. 
\end{abstract}

\section{Introduction}
Consider
the basic setting of a single seller that is selling multiple goods to a single
buyer in a Bayesian setup, where only the probability distribution of the
buyer's valuations is known to the seller. 
What is the optimal 
mechanism that
maximizes the seller's expected revenue from this distribution? In contrast to the
single-good case that was fully solved by \cite{Mye81},
this turns out to be a difficult problem due to the
``multi-parameter'' nature of the problem.
See, e.g., \cite{BCKW10, CHMS10, DW11, Tha04, MM88, MV06, HN12, HN13, DDT13, DDT15, HR15, BGN17}, among many others.

This paper deals with a more preliminary question: do optimal
mechanisms exist at all? (The alternative would be to have 
mechanisms
that can extract higher and higher revenues, but never achieve the maximal limit revenue.) Having such a revenue-maximizing mechanism allows us to simplify various arguments and dispense with constructs that start with 
the annoying
``Let $\varepsilon>0$ and let $\mu$ be a mechanism that 
extracts a revenue of at least $\Rev(X)-\varepsilon$ from $X$.''

The following example shows that an optimal mechanism need not always exist, even in the case of a single good.  Assume that the valuation of the good is given by a random variable $X$ with
$\mathbb{P}\left[X\geq t\right]  =1/(t+1)$ for every $t\geq0$ (i.e., with density $1/(t+1)^{2}$). The
revenue that can be obtained by the fixed price $p$ is thus $p \cdot \mathbb{P}\left[X\geq p\right]=p/(p+1)$, and so, by \cite{Mye81}, the optimal revenue is $\Rev(X)=\sup_{p\geq0} p/(p+1) =1,$ but there is no finite price $p,$ and thus no mechanism (which
is a convex combination of fixed price mechanisms) where the revenue 1 is achieved.\footnote{A discrete version of the example: for every integer $n \ge 0$ let $\mathbb{P}\left[  X\geq n\right]  =1/(n+1),$ i.e.,
$\mathbb{P}\left[  X=n\right]  =1/((n+1)(n+2))$.}

For multiple goods, the elegant but complex duality analysis of
\cite{DDT15} shows that optimal mechanisms exist when the valuations are bounded
and the probability distributions have continuous densities that are differentiable and have bounded derivatives.

In this note we provide a simple elementary proof of existence of optimal mechanisms under the very minimal
condition that the random valuation has \emph{finite expectation}
(i.e., is integrable).  

The proof strategy is what one would expect: showing that a
``limit'' of mechanisms is itself a mechanism.  The question is how to define such a limit 
properly.  Directly looking at the limit of allocations and payments does not seem
to do the trick.  As we will show, what does work is looking at the (pointwise) limit of the \emph{buyer
payoff functions}.

We use the following rather general formalization to state our results.  We denote by $\Gamma \subset \mathbb{R}_+^k$ the set of possible ``allocations'' to the buyer, where 
$\Gamma$ can be any compact (bounded and closed) set of nonnegative $k$-dimensional vectors.\footnote{We write $\mathbb{R}_+$ for $\mathbb{R}_{\ge 0}=\{x:x \ge 0\}$.} A buyer's valuation is 
given by another $k$-dimensional nonnegative vector 
$x \in \mathbb{R}_{+}^k$, 
which yields a real value of
$g \cdot x = \sum_{i=1}^k g_i x_i$ for each possible allocation $g \in \Gamma$.  This formalization
directly models mechanisms for $k$ goods with additive valuation and also with unit demand, and abstract implementation with $k$ choices, both for deterministic mechanisms and for general (randomized) mechanisms; see Table \ref{tab:gamma_choices}.  Most other settings (such as combinatorial 
valuations) are easily reduced to one of these, with $k$ being the appropriate number of parameters (for combinatorial auctions, 
$k$ is exponential in the number of goods).

\begin{table}[h]
    \centering
    \renewcommand{\arraystretch}{1.5} 
    
    \begin{tabularx}{\textwidth}{|l|X|X|} 
        \hline
        \textbf{Setting} & \textbf{Deterministic Mechanisms} & \textbf{Randomized Mechanisms} \\ 
        \hline
        One good & \centering $\Gamma=\{0,1\}$ & \centering $\Gamma=[0,1]$ \tabularnewline 
        \hline
        $k$ goods with additive valuation & \centering $\Gamma=\{0,1\}^k$ & \centering $\Gamma=[0,1]^k$ \tabularnewline 
        \hline
        $k$ goods with unit demand & \centering $\Gamma=\{\mathbf{0},e^{1},...,e^{k}\}$ & \centering $\Gamma=\{g\in\lbrack
0,1]^{k}:\sum_{i}g_{i}\leq 1\}$ \tabularnewline 
        \hline
        Implementation with $k$ options & \centering $\Gamma=\{e^{1},...,e^{k}\}$ & \centering $\Gamma=
        \{g\in\lbrack0,1]^{k}:\sum_{i}g_{i}=1\}$ \tabularnewline 
        \hline
    \end{tabularx}
    
    \caption{Some choices of $\Gamma$.  We denote by $e^i$ the unit vector in direction $i$ and by $\mathbf{0}$ the all-0 vector.}
    \label{tab:gamma_choices}
\end{table}

A \emph{mechanism} $\mu$ in this setting
consists of two functions, the
\emph{allocation function }$q:\mathbb{R}_{+}^{k}\rightarrow\Gamma$ and the
\emph{payment function} $s:\mathbb{R}_{+}^{k}\rightarrow\mathbb{R}$. 
We require our mechanisms to be
both \emph{incentive compatible} (\emph{IC}), i.e.,
$q(x)\cdot x-s(x)\geq q(y)\cdot x-s(y)$
for every $x$ and $y$ in $\mathbb{R}_{+}^{k}$, and \emph{individually
rational} (\emph{IR}), i.e.,
$q(x)\cdot x-s(x)\geq 0$
for every $x$ in $\mathbb{R}_{+}^{k}.$

We consider the Bayesian setting where the buyer's valuation is given by a random variable $X$ with values in $\mathbb{R}_{+}^{k}$; the seller knows only the distribution of $X$ (we refer to $X$ as a 
\emph{random valuation}).  
The revenue that a mechanism $\mu$ extracts from
$X$ 
is the expected payment, $R(\mu;X):=\mathbb{E}\left[  s(X)\right]  ,$ and the \emph{optimal
revenue} that can be extracted from $X$ is
$\Rev_\Gamma(X):=\sup_\mu R(\mu;X)$,
where the supremum is taken over all (IC and IR) mechanisms $\mu$.  We can now state our theorem.

\begin{theorem}
\label{th:exist}For every compact set of possible allocations $\Gamma$ and every $k$-good random valuation $X$ with finite expectation there exists a
revenue-maximizing mechanism $\mu$, i.e., 
$R(\mu;X)=\Rev_\Gamma(X)$.
\end{theorem}

It thus follows (from Table \ref{tab:gamma_choices}) that $\Rev$ is attained for $k$ goods, in the additive case as well as in  the unit-demand case; the same holds for $\textsc{DRev}$, the revenue by deterministic mechanisms. For the bundled revenue $\textsc{BRev}$ and the separate revenue $\textsc{SRev}$, this follows from the single-good case.\footnote{For $\textsc{BRev}$ this also follows by taking $\Gamma=\{\mathbf{0},\mathbf{1}\}$, where $\mathbf{1}$ denotes the all-1 vector.} In the Appendix we show how our construct yields existence for additional subclasses of mechanisms: monotonic and allocation-monotonic mechanisms.

\section{The Model\label{s:model}}

Let $k\geq1$ be the dimension. The domain of
\emph{valuations} is $\mathbb{R}_{+}^{k}$, the nonnegative orthant of
$\mathbb{R}^{k},$ and the set \emph{allocations }is a nonempty compact set
$\Gamma\subset\mathbb{R}_{+}^{k}$ (such as  the unit cube, the unit simplex, or their vertices).\footnote{The allocations $g$ and the valuations $x$ belong to \emph{dual} $\mathbb{R}^k$ spaces, both conveniently endowed with the standard Euclidean norm (we do not need precise bounds here, and so do not use more appropriate norms as in \cite{HN25}).}
\ignore{In the leading examples $\Gamma$ is the
unit cube or the unit simplex.}
Let 
$\gamma:=\max_{g\in\Gamma}\left\Vert
g\right\Vert .$

\subsection{Mechanisms}

A (direct) $\Gamma$-\emph{mechanism} $\mu$ consists of two functions, the
\emph{allocation function }$q:\mathbb{R}_{+}^{k}\rightarrow\Gamma$ and the
\emph{payment function} $s:\mathbb{R}_{+}^{k}\rightarrow\mathbb{R}$. A
mechanism $\mu$ is \emph{incentive compatible} (\emph{IC}) if%
\[
q(x)\cdot x-s(x)\geq q(y)\cdot x-s(y)
\]
for every $x$ and $y$ in $\mathbb{R}_{+}^{k}$; and it is \emph{individually
rational} (\emph{IR}) if%
\[
q(x)\cdot x-s(x)\geq0
\]
for every $x$ in $\mathbb{R}_{+}^{k}.$ Thus, when the buyer's valuation (or
type) is $x,$ his payoff is
\begin{equation}
b(x):=q(x)\cdot x-s(x),\label{eq:b}%
\end{equation}
and the mechanism's payoff (or revenue) is $s(x).$ Individual rationality (IR)
requires that 
\[
b(x)\geq0
\]
for every $x\in\mathbb{R}_{+}^{k},$ and incentive
compatibility (IC) that
\begin{equation}
b(x)=\max_{z\in\mathbb{R}_{+}^{k}}[q(z)\cdot x-s(z)]\label{eq:b-max}
\end{equation}
for every $x\in\mathbb{R}_{+}^{k}.$ Hereafter we will write a mechanism
as\footnote{While $b$ is fully determined by $q$ and $s,$ it is convenient for
the statements below to have $b$ included in $\mu$ as well (rather than saying
``a mechanism $\mu$ with buyer payoff function $b$'').}
$\mu=(q,s,b).$

\subsection{Revenue}

A \emph{random valuation} $X$ is a random variable with values in
$\mathbb{R}_{+}^{k}.$ The revenue that a mechanism $\mu=(q,s,b)$ extracts from
$X$ is $R(\mu;X):=\mathbb{E}\left[  s(X)\right]  ,$ and the \emph{optimal
revenue} that can be extracted from $X$ is
\[
\text{\textsc{Rev}}_{\Gamma}(X)%
{\;:=\;}%
\sup R(\mu;X),
\]
where the supremum is taken over all IC and IR $\Gamma$-mechanisms $\mu$.

When maximizing revenue it suffices to consider only those IC and IR
mechanisms that satisfy the \emph{no positive transfer} (\emph{NPT})
property: $s(x) \ge 0$ for every $x$. Indeed, if the minimal payment, which is  $s(\mathbf{0})$ (by IC at $\mathbf{0}$), is negative, then increasing all payments by $|s(\mathbf{0})|$ preserves IC and IR and increases the revenue. Moreover, since $b(\mathbf{0})=-s(\mathbf{0})$, for IR mechanisms NPT is equivalent to $s(\mathbf{0})=0$, and thus to $b(\mathbf{0})=0$ (cf. Proposition 6 in
\cite{HN12}). 

Let $\mathcal{M}_{\Gamma}$ denote the set of all
IC, IR, and NPT $\Gamma$-mechanisms.

\subsection{Buyer Payoff Functions}

To avoid having to deal with inessential technical issues on the boundary of\footnote{For instance, subgradients with (arbitrarily large) negative coordinates at boundary points.} $\mathbb{R}_{+}^{k}$, it is
convenient to extend the buyer payoff function $b$ to
an open neighborhood of $\mathbb{R}_{+}^{k}$, in fact to the entire space $\mathbb{R}^{k}$ (cf. the Appendix of \cite{HR15}), by%
\begin{equation}
b(x):=\sup_{z\in\mathbb{R}_{+}^{k}}[q(z)\cdot x-s(z)] \label{eq:b-ext}%
\end{equation}
for every $x\in\mathbb{R}^{k}$ (i.e., by extending (\ref{eq:b-max})).  The resulting function $b$ is well defined and
finite for every $x\in\mathbb{R}^{k}$, because for each $z\in\mathbb{R}%
_{+}^{k}$ the function $q(z)\cdot x-s(z)$ is $\gamma$-Lipschitz in $x$ (recall
that $\gamma=\max_{g\in\Gamma}\left\Vert g\right\Vert $), and thus so is the
supremum of these functions, $b$: for every $x,y\in\mathbb{R}^{k}$ we have $\left\vert
b(x)-b(y)\right\vert \leq\gamma\left\Vert x-y\right\Vert $. Hereafter $b$
will always stand for this extended function $b:\mathbb{R}^{k} \rightarrow
\mathbb{R}$ given by (\ref{eq:b-ext}).

We recall now a few basic concepts for convex functions (see \cite{Rock}, Sections
23--25; for the convergence results, see in particular Theorems 24.5, 24.6, and 25.6 there). 
Let $f:\mathbb{R}^{k}\rightarrow\mathbb{R}$ be a real convex function
defined on $\mathbb{R}^{k}$ (i.e., $\mathrm{dom}~f=\mathbb{R}^{k}$). A vector
$g\in\mathbb{R}^{k}$ is a \emph{subgradient} of $f$ at $x\in\mathbb{R}^{k}$ if
$f(y)-f(x)\geq g\cdot(y-x)$ for all $y\in\mathbb{R}^{k}$. The set of
subgradients of $f$ at $x,$ denoted by $\partial f(x)$, is a nonempty convex and compact
set. When $b$ is differentiable at $x$, which holds almost everywhere, the
unique subgradient is the gradient, i.e., $\partial b(x)=\{\nabla b(x)\}$; let
$D$ denote the set of points where $b$ is differentiable. The
\emph{directional derivative of }$f$ \emph{at} $x\in\mathbb{R}^{k}$ \emph{in
the direction} $y\in\mathbb{R}^{k}$ is $f^{\prime}(x;y):=\lim_{\delta
\rightarrow0^{+}}(f(x+\delta y)-f(x))/\delta$. It always exists, and
$f^{\prime}(x;y)=\max\{g\cdot y:g\in\partial f(x)\}$; let $\partial
f(x)_{y}:=\{g\in\partial f(x):g\cdot y=f^{\prime}(x;y)\}$ denote the set of
maximizers. Let $x_{n}\rightarrow x$; if $g_{n}$ is a subgradient of $f$ at
$x_{n}$, i.e., $g_{n}\in\partial f(x_{n})$, and $g_{n}\rightarrow g$ then $g$
is a subgradient of $f$ at $x$, i.e., $g\in\partial f(x)$; moreover, if
$x_{n}\rightarrow x$ from the direction $y$, i.e., $(x_{n}-x)/\left\Vert
x_{n}-x\right\Vert \rightarrow y$, then the subgradient $g$ is maximal in the
direction $y,$ i.e., $g\in\partial f(x)_{y}$. 
Finally, the set of subgradients $\partial b(x)$ is the closed convex hull of the set of all limit points of sequences of gradients $\nabla b(x_n)$, where $x_n$ is a sequence in $D$ converging to $x$. 

\bigskip

Let $\mathcal{B}_{\Gamma}$ denote the set of all convex functions
$b:\mathbb{R}^{k}\rightarrow\mathbb{R}$ with $b(\mathbf{0})=0$ and subgradients
in $\Gamma$, by which we mean that at every $x$ in $\mathbb{R}^{k}$ there is a subgradient in $\Gamma$, i.e., $\partial b(x)\cap\Gamma\neq\emptyset$. Since $\Gamma$ is a compact set, it suffices to require that $\nabla b(x)\in\Gamma$ for every $x\in\mathbb{R}^{k}$ where
$b$ is differentiable, i.e., $x\in D$. 
Indeed, take a sequence of points $x_{n}$
in $D$ converging to $x$; the gradients $\nabla b(x_{n})$ are all in the
compact set $\Gamma$, and so any limit point of the sequence $\nabla b(x_{n})$---which is a subgradient at the limit point $x$---is also in\footnote{When the set $\Gamma$ is in addition convex, \emph{all}
subgradients belong to $\Gamma,$ i.e., $\partial
b(x)\subseteq\Gamma$ for every $x\in\mathbb{R}^{k}$ (by Theorem 25.6 in
\cite{Rock}).} $\Gamma$. Moreover, by taking $x_{n}\in D$ so that it converges to $x$ from the
direction $y$---for instance, take $x_{n}$ in $D$ to be within a distance of
$1/n^{2}$ from $x+(1/n)y$---we obtain $\partial b(x)_{y}\cap\Gamma\neq
\emptyset$ for every $x$ and $y$ in $\mathbb{R}^{k}$. Finally, the inequality $b(y) \ge b(x) +g \cdot (y-x)$ with $g \in \partial b(x) \cap \Gamma$ gives $b(x)-b(y) \le \gamma \Vert x-y \Vert$,  and so every function $b$ in $\mathcal{B}_\Gamma$ is $\gamma$-Lipschitz;  together with $b(\mathbf{0})=0$, it follows that
\begin{equation}
|b(x)| \le \gamma\left\Vert x\right\Vert
\label{eq:bdd}%
\end{equation}
for every $x$ in $\mathbb{R}^{k}.$

\subsection{Buyer Payoff Functions and Mechanisms}

The following is a classic result (see \cite{Roch}, \cite{HN12}), restated for
general $\Gamma$-mechanisms.

\begin{proposition}
\label{p:M_Gamma}Let $\mu=(q,s,b)$ be a $\Gamma$-mechanism. Then $\mu$ is in
$\mathcal{M}_{\Gamma}$ if and only if the function $b$ is in $\mathcal{B}%
_{\Gamma}$ and, for every $x\in\mathbb{R}_{+}^{k}$, the vector $q(x)\in\Gamma$
is a subgradient of $b$ at $x$, i.e., $q(x)\in\partial b(x)\cap\Gamma$.
\end{proposition}

\begin{proof}
If $\mu$ is in $\mathcal{M}_{\Gamma}$ then $b$ is a convex function (as the
supremum of affine functions), and satisfies $b(\mathbf{0})=0$ (by IR and NPT). For every $x\in\mathbb{R}_{+}^{k}$ the vector $q(x)\in\Gamma$ is a
subgradient of $b$ at $x$, because for every $y\in\mathbb{R}^{k}$ we have
$b(y)\geq q(x)\cdot y-s(x)=b(x)+q(x)\cdot(y-x)$ (by IC). For $x$ outside $\mathbb{R}_+^k$, by the compactness of $\Gamma$ there is $(g,t) \in \Gamma \times \mathbb{R}$ in the closure of $\{(q(z),s(z)):z \in \mathbb{R}^{k}_+\}$  (the ``menu'' of $\mu$) such that $b(x)=g \cdot x -t$, and then, as above, $b(y)\geq g\cdot y-t=b(x)+g\cdot(y-x)$, and thus $g \in \Gamma$ is a subgradient of $b$ at $x$. Therefore $\partial b(x)\cap\Gamma \neq \emptyset$ for every $x \in \mathbb{R}^k$, and so $b \in \mathcal{B}_\Gamma$.

Conversely, if $b\in\mathcal{B}_{\Gamma}$ then for each
$x\in\mathbb{R}_{+}^{k}$ choose $q(x)\in\partial b(x)\cap\Gamma\neq\emptyset$
and put $s(x):=q(x)\cdot x-b(x);$ then IR holds because for every
$x\in\mathbb{R}_{+}^{k}$ we have $b(x)\geq b(\mathbf{0})+q(\mathbf{0}%
)\cdot(x-0)\geq0$ (use $q(\mathbf{0})\in\partial
b(\mathbf{0})$, $b(\mathbf{0})=0$, and $q(\mathbf{0})\in\Gamma\subset\mathbb{R}_{+}^{k}$), IC
because for every $x,y\in\mathbb{R}_{+}^{k}$ we have $q(y)\cdot
x-s(y)=b(y)+q(y)\cdot(x-y)\leq b(x)$ (the inequality because $q(y)\in\partial
b(y)$), and NPT because $b(\mathbf{0})=0$. \ignore{$\geq b(x)+q(x)\cdot
(\mathbf{0}-x)=-s(x)$ (the inequality because $q(\mathbf{0)\in\partial
}b(\mathbf{0)}$).}
\end{proof}

\bigskip

Next, we see how the payments are determined by the buyer payoff function (see \cite{HR15}).

\begin{proposition}
\label{p:sf}Let $b\in\mathcal{B}_{\Gamma}$. For every $\mu=(q,s,b)$ in
$\mathcal{M}_{\Gamma}$ we have $s(x)\leq b^{\prime}(x;x)-b(x)$ for every
$x\in\mathbb{R}_{+}^{k},$ and there is a mechanism $\mu^{\ast}=(q^{\ast
},s^{\ast},b)$ in $\mathcal{M}_{\Gamma}$ with $s^{\ast}(x)=b^{\prime
}(x;x)-b(x)$ for every $x\in\mathbb{R}_{+}^{k}$.
\end{proposition}

\begin{proof}
Since $q(x)\in\partial b(x)$ by Proposition \ref{p:M_Gamma}, we get
$s(x)=q(x)\cdot x-b(x)\leq b^{\prime}(x;x)-b(x)$. When constructing $\mu$ from
$b$ in the proof of Proposition \ref{p:M_Gamma} we can choose $q^{\ast}(x)$ to be moreover maximal in the direction $x$,
i.e., $q^{\ast}(x)\in\partial b(x)_{x}\cap\Gamma\neq\emptyset$, for each
$x\in\mathbb{R}_{+}^{k}$; then $q^{\ast}(x)\cdot x=b^{\prime}(x;x)$, and so
$s^{\ast}(x):=q^{\ast}(x)\cdot x-b(x)=b^{\prime}(x;x)-b(x)$.
\end{proof}

\bigskip

The mechanism $\mu^{\ast}$ of Proposition \ref{p:sf}, called a
\emph{seller-favorable} mechanism in \cite{HR15}, yields to the seller the
highest payments obtainable from all mechanisms with the same buyer payoff function $b$
(it amounts to the buyer, when indifferent, breaking ties in favor of the
seller); when maximizing revenue, the seller-favorable mechanisms are the only ones that matter. Thus, $\Rev_\Gamma (X) = \sup \mathbb{E}[b'(X;X)-b(X)]$, where the supremum is taken over all $b \in \mathcal{B}_\Gamma$. 

\bigskip

\section{Proof}\label{s:proof}

The proof consists in showing, first, 
that the set of buyer payoff functions is sequentially compact with respect to
pointwise convergence (see Proposition \ref{l:subseq} below), and second, that
the revenue is upper semicontinuous with
respect to this convergence (see Proposition \ref{p:bn->b} below, which uses the integrability of the valuation\footnote{As shown by the example in the Introduction, an optimal mechanism need not exist otherwise.}).

\begin{proposition}
\label{l:subseq}Let $b_{n},$ for $n=1,2,...$, be a sequence of functions in
$\mathcal{B}_{\Gamma}.$ Then there exists a subsequence $n^{\prime},$ w.l.o.g.~the
original sequence $n$, such that $b_{n}$ converges pointwise to a limit
function $b$, i.e., $\lim_{n\rightarrow\infty}b_{n}(x)=b(x)$ for every
$x\in\mathbb{R}^{k}$, and the function $b$ is in $\mathcal{B}_{\Gamma}$.
\end{proposition}

\begin{proof}
For each $x$ the sequence $(b_{n}(x))_{n\geq1}$ is bounded (by $\gamma \left\Vert x\right\Vert$; see (\ref{eq:bdd})),
and so Theorem 10.9 of \cite{Rock} gives the result.\footnote{The construction
is standard (cf.~the Helly selection theorem, and the Arzel\`{a}--Ascoli theorem, which
suffices for bounded domains of valuations): take a countable dense set of
points for which we obtain a sequence of convergent subsequences, then use the
``diagonal'' subsequence, and apply
continuity.} By Theorem 24.5 in \cite{Rock}, 
the sets $\partial b_{n}(x)$ converge to the set $\partial b(x),$
and so $\partial b_{n}(x)\cap\Gamma\neq\emptyset$ implies $\partial
b(x)\cap\Gamma\neq\emptyset$ (because $\Gamma$ is compact). Together with
$b(\mathbf{0})=\lim_{n}b_{n}(\mathbf{0})=0$, we get $b\in\mathcal{B}_{\Gamma}.$
\end{proof}

\begin{proposition}
\label{p:bn->b}Let $\mu_{n}=(q_{n},s_{n},b_{n}),$ for $n=1,2,...$, 
and $\mu=(q,s,b)$ be mechanisms in $\mathcal{M}_{\Gamma}$. If $b_{n}$ converges
pointwise to $b$ and $\mu$ is seller favorable, then 
\[
\limsup_{n\rightarrow\infty}s_{n}(x)\leq s(x)
\]
for
every $x$, and thus 
\[
\limsup_{n\rightarrow\infty}R(\mu_{n};X)\leq R(\mu;X)
\]
for every random valuation $X$ with finite expectation.
\end{proposition}

\begin{proof}
For every $x$ in $\mathbb{R}_{+}^{k}$, we have%
\[
\limsup_{n\rightarrow\infty}s_{n}(x)\leq\limsup_{n\rightarrow\infty} \left[
b_{n}^{\prime}(x;x)-b_{n}(x)\right]  \leq b^{\prime}(x;x)-b(x)=s(x)
\]
(the first inequality by Proposition \ref{p:sf}, the second because $\limsup_{n}
b_{n}^{\prime}(x;x) \leq b^{\prime}(x;x)$ by Theorem 24.5 in \cite{Rock}, and the final equality because $\mu$ is seller favorable). 

Next, 
\[
0\leq s_{n}(x)\leq q_{n}(x)\cdot x\leq\left\Vert q_{n}%
(x)\right\Vert \left\Vert x\right\Vert \leq\gamma\left\Vert x\right\Vert 
\]
(the first two inequalities by NPT and IR) for
every $x$ in $\mathbb{R}_{+}^{k}$ and $n\geq1$, and thus, for an integrable $X$,
the sequence $s_{n}(X)$ is dominated by the integrable function $\gamma
\left\Vert X\right\Vert $. Therefore,%
\begin{align*}
\limsup_{n\rightarrow\infty}R(\mu_{n};X)  & =\limsup_{n\rightarrow\infty
}\mathbb{E}\left[  s_{n}(X)\right] \\
& \leq\mathbb{E}\left[  \limsup_{n\rightarrow\infty}s_{n}(X)\right]
\leq\mathbb{E}\left[  s(X)\right]  =R(\mu;X),
\end{align*}
where the first inequality is by Fatou's lemma applied to the sequence of
nonnegative functions $\gamma\left\Vert X\right\Vert -s_{n}(X)$.
\end{proof}

\bigskip

This proves our result:

\bigskip

\begin{proof}
[Proof of Theorem \ref{th:exist}]Let $\mu_{n}=(q_{n},s_{n},b_{n}),$ for $n \ge 1$, be a sequence of
mechanisms in $\mathcal{M}_{\Gamma}$ such that $R(\mu_{n};X)\rightarrow_{n}%
\,$\textsc{Rev}$_{\Gamma}(X)$;  thus, $b_{n}\in\mathcal{B}_{\Gamma}$ by Proposition \ref{p:M_Gamma}. Proposition
\ref{l:subseq} then yields $b\in\mathcal{B}_{\Gamma}$ and a subsequence $n^{\prime
}$, which w.l.o.g.~we take to be the original sequence $n$, such that
$b_{n}$ converges pointwise to $b$. 
Next, Proposition \ref{p:sf} provides a
seller-favorable mechanism $\mu=(q,s,b)$ in $\mathcal{M}_{\Gamma}$ with
$s(x)=b^{\prime}(x;x)-b(x)$ for every $x$ in $\mathbb{R}_{+}^{k}.$ Finally,
$R(\mu;X)\geq\lim_{n}R(\mu_{n};X)=\,$\textsc{Rev}$_{\Gamma}(X)$ by Proposition
\ref{p:bn->b}, with equality since $\mu$ is in $\mathcal{M}_{\Gamma}$.
\end{proof}

\bigskip

\appendix{}

\section{Appendix: Subclasses of Mechanisms}

Does the existence result extend to subclasses of mechanisms? As we have seen, the answer
is immediately positive, by Theorem \ref{th:exist}, when the subclass corresponds to a certain compact set of allocations $\Gamma$ (as is the case, for instance, for deterministic mechanisms, where $\Gamma = \{0,1\}^k$).  However, our above proof applies to any subclass of mechanisms that is closed under the pointwise convergence
of the buyer payoff functions, i.e., provided that there is a ``limit'' mechanism $\mu$ in Proposition
\ref{p:bn->b} that is in the same subclass as the sequence of mechanisms $\mu_n$.

We provide here the result for two interesting such subclasses:  monotonic mechanisms and allocation-monotonic mechanisms (see \cite{HR15, BHN}).

\subsection{Monotonic Mechanisms}

A mechanism $\mu=(q,s,b)\ $is \emph{monotonic} if $s(y)\geq s(x)$ for every
$y\geq x$ in $\mathbb{R}_{+}^{k}.$ Let $\textsc{MonRev}_{\Gamma}(X)$
denote the maximal revenue that can be extracted from $X$ by monotonic
$\Gamma$-mechanisms.
The result is:

\begin{theorem}
\label{th:mon}For every compact set of possible allocations $\Gamma$ and every $k$-good random valuation $X$ with finite expectation there exists a monotonic
revenue-maximizing mechanism
$\mu$, i.e.,
$R(\mu;X)=
\normalfont\textsc{MonRev}_{\Gamma}(X).$
\end{theorem}

Again, the result applies to the additive-valuation setup as well as the unit-demand setup, for general (randomized) mechanisms, and also for deterministic mechanisms. The proof, as in Section \ref{s:proof}, uses the following additional result.

\begin{proposition}
\label{p:bn->b monot}Let $\mu_{n}=(q_{n},s_{n},b_{n}),$ for $n=1,2,...$, 
and $\mu=(q,s,b)$ be 
$\Gamma$-mechanisms in $\mathcal{M}_{\Gamma}$. If all the $\mu_{n}$ are
monotonic, $b_{n}$ converges pointwise to $b$, and $\mu$ is seller favorable, then $\mu$ is
monotonic as well.
\end{proposition}

\begin{proof}
Let $y\geq x$ be two points in $\mathbb{R}_{+}^{k};$ we need to show that
$s(y)\geq s(x)$. 

(i) Assume first that $x$ is in $D$ (the dense set of points where $b$ is differentiable), and so $q(x)=\nabla b(x).$ By Theorem
24.5 in \cite{Rock}, we get $q_{n}(x)\rightarrow_{n}q(x),$ and so
$s_{n}(x)=q_{n}(x)\cdot x-b_{n}(x)\rightarrow_{n}q(x)\cdot x-b(x)=s(x)$. Now $s_{n}(y)\geq s_{n}(x)$ for every $n$ (since
the $\mu_{n}$ are monotonic), and so, by
Proposition \ref{l:subseq}, it follows that $s(y)\geq\limsup_{n}s_{n}(y)\geq\lim_{n}%
s_{n}(x)=s(x)$.

(ii) For a general $x \in \mathbb{R}^k_+$ (not necessarily in $D$), we proceed as follows. Let
$x^{m}$ be a sequence of points in $D$ such that $x^m \ge x$ and $x^{m}\rightarrow_{m}x$ from
the direction $x$; then $q(x^{m})=\nabla b(x^{m})\rightarrow_{m}\partial
b(x)_{x}$. 
Since $g\cdot x=b^{\prime}(x;x)$ for
every $g\in\partial b(x)_{x},$ it follows that $s(x^{m})=q(x^{m})\cdot
x^{m}-b(x^{m})\rightarrow_{m}b^{\prime}(x;x)-b(x)=s(x).$ Let $y^{m}:=y+x^{m}-x$; then $y^{m}\rightarrow_{m}y$ and $y^{m}\geq x^{m}\in D,$ and so
$s(y^{m})\geq s(x^{m})$ by (i) above. The function $s$ is upper
semicontinuous (because $b^{\prime}$ is upper semicontinuous and $b$ is continuous; see
Theorem 10.1 and Corollary 24.5.1 in \cite{Rock}), and so $s(y)\geq\limsup_{m}s(y^{m})\geq\lim_{m}s(x^{m})=s(x).$

Thus $s(y)\geq s(x)$ in both cases, completing the proof.
\end{proof}

\bigskip

We note that the limit $\mu$ need not be monotonic when it is not seller favorable (just break the tie at some point in the ``wrong way'').

\subsection{Allocation-Monotonic Mechanisms}

A mechanism $\mu=(q,s,b)$ is \emph{allocation monotonic} if $q(y)\geq
q(x)$ for every $y\geq x$ in $\mathbb{R}_{+}^{k}.$ Let $\textsc{AMonRev}_{\Gamma}(X)$ denote the maximal revenue that can be extracted from $X$ by
allocation-monotonic $\Gamma$-mechanisms.

\begin{theorem}
\label{th:amon}For every compact set of possible allocations $\Gamma$ and every $k$-good random valuation $X$ with finite expectation there exists an allocation-monotonic
revenue-maximizing mechanism
$\mu$, i.e.,
$R(\mu;X)=
\normalfont\textsc{AMonRev}_{\Gamma}(X).$
\end{theorem}

For the proof we use:

\begin{proposition}
\label{p:bn->b amon}Let $\mu_{n}=(q_{n},s_{n},b_{n})$, for $n=1,2,...$, 
and 
$\mu=(q,s,b)$ be 
$\Gamma$-mechanisms in $\mathcal{M}_{\Gamma}$. If all the $\mu_{n}$ are
allocation monotonic, $b_{n}$ converges pointwise to $b$, and $\mu$ is tie favorable (i.e., seller
favorable as well as buyer favorable), then $\mu$ is allocation monotonic as well.
\end{proposition}

\begin{proof}
In \cite{BHN} (Theorem C, Proposition 4.1, and Appendix A-6), it is shown that, for tie-favorable mechanisms, $\mu$ is
allocation monotonic if and only if $b$ is a supermodular function on
$\mathbb{R}_{+}^{k}.$ The supermodular inequalities are clearly preserved when
taking limits: if $b_{n}\rightarrow b$ pointwise and $b_{n}$ is supermodular
for each $n$, then so is $b$. For supermodular functions, at every point there is a coordinatewise-maximal subgradient $q^*(x)$, which must be in $\Gamma$ (this is seen by taking points $x^m$ in $D$ that converge to $x$ from the direction $(1,1,...,1)$),
and so the unique tie-favorable mechanism for this $b$, which uses $q^*$, is allocation monotonic and in $\mathcal{M}_\Gamma$.
\end{proof}

\bibliographystyle{alpha}
\bibliography{bib}

\end{document}